\documentclass[11pt]{amsart}

\usepackage{amsmath}
\usepackage{amssymb}
\usepackage{amsthm}
\usepackage{graphicx}
\usepackage{tikz}
\usepackage{xcolor}

\allowdisplaybreaks

\usepackage[margin=1.2in]{geometry}
\usepackage{bookmark}

\newtheorem{theorem}{Theorem}
\newtheorem*{theorem*}{Theorem}
\newtheorem{corollary}[theorem]{Corollary}
\newtheorem{lemma}[theorem]{Lemma}
\newtheorem{question}{Question}
\newtheorem{remark}{Remark}

\newtheorem{definition}[theorem]{Definition}

\newtheorem*{claim*}{Claim}

\newcommand{\eps}{\epsilon}
\newcommand{\lam}{\lambda}
\newcommand{\cI}{\mathcal{I}}

\begin{document}

\title[Approximately counting independent sets in bipartite graphs via containers]{Approximately counting independent sets in bipartite graphs via graph containers}

\author{Matthew Jenssen}
\address{School of Mathematics\\ University of Birmingham}
\author{Will Perkins}
\address{Department of Mathematics, Statistics, and Computer Science\\ University of Illinois at Chicago}
\author{Aditya Potukuchi}
	\email{m.jenssen@bham.ac.uk \\ math@willperkins.org \\adityap@uic.edu }

\date{\today}

\date{\today}

\maketitle

\begin{abstract}
By implementing algorithmic versions of Sapozhenko's graph container methods, we give new algorithms for approximating the  number of independent sets in bipartite graphs.  Our first algorithm applies to $d$-regular, bipartite graphs satisfying a weak expansion condition: when $d$ is constant, and the graph is a bipartite $\Omega( \log^2 d/d)$-expander, we obtain an FPTAS for the number of independent sets.  Previously such a result for $d>5$ was known only for graphs satisfying the much stronger expansion conditions of random bipartite graphs. The  algorithm also applies to weighted independent sets: for a $d$-regular, bipartite $\alpha$-expander, with $\alpha>0$ fixed, we give an FPTAS for the hard-core model partition function at fugacity $\lambda=\Omega(\log d / d^{1/4})$. Finally we present an algorithm that applies to all $d$-regular, bipartite graphs, runs in time $\exp\left(  O\left(  n \cdot \frac{ \log^3 d }{d } \right) \right)$, and outputs a $(1 + o(1))$-approximation to the number of independent sets.
\end{abstract}

\section{Introduction}

Let $\cI(G)$ denote the set of independent sets of a graph $G$ and $i(G) = |\cI(G)|$ denote the number of independent sets of $G$.    Computing $i(G)$ is  a \#P-hard problem, even when restricted to bounded degree, bipartite graphs~\cite{provan1983complexity}.    Even approximating $i(G)$ (to a constant or even subexponential factor) remains NP-hard, even when restricted to $d$-regular graphs with $d \ge 6$~\cite{sly2010computational}.    

Intuitively, one might expect the problem of approximating $i(G)$ to be easier on the class of bipartite graphs; for one, there is a polynomial-time algorithm to find a maximum-size independent set in  a bipartite graph while the corresponding problem is NP-hard for general graphs.  Dyer, Goldberg, Greenhill, and Jerrum~\cite{dyer2004relative} defined the counting problem \#BIS (bipartite independent set) and showed that several natural combinatorial counting problems are as hard to approximate as \#BIS.  These problems include counting  stable matchings, approximating the ferromagnetic Potts model partition function ($q \ge 3$)~\cite{goldberg2012approximating,galanis2016ferromagnetic},  counting the number of $q$-colorings in a bipartite graph ($q \ge 3$), and approximating the ferromagnetic Ising model partition with non-uniform external fields~\cite{goldberg2007complexity}.  

The search for approximation algorithms for $i(G)$ that exploit bipartite structure generally falls into two categories.  The first approach finds classes of graphs on which polynomial-time approximation algorithms exist.  Liu and Lu~\cite{liu2015fptas} gave the first such algorithm, providing an FPTAS for the number of independent sets in bipartite graphs in which degrees on one side are bounded by $5$, while degrees on the other side are unrestricted.  Another line of work in this direction includes~\cite{jenssen2020algorithms} in which approximation algorithms are given for the hard-core partition function $Z_G(\lam)$ (counting weighted independent sets) in bounded-degree, bipartite expander graphs, based on two tools from statistical physics: polymer models and the cluster expansion (following~\cite{helmuth2020algorithmic}).  This work was followed by several improvements, extensions, and generalizations including~\cite{liao2019counting,cannon2020counting,chen2021fast,galanis2020fast,friedrich2020polymer,galanis2021fast,chen2021sampling}.  All of these algorithms are `low-temperature' algorithms: they exploit the fact that on a bipartite graph with sufficient expansion, most (weighted) independent sets have few vertices from one side of the bipartition; that is, they are close to one of the two ground states consisting of all subsets of one side.  In contrast, the algorithms of Weitz~\cite{weitz2006counting} (an FPTAS for $i(G)$ in graphs of maximum degree at most $5$) and Liu and Lu~\cite{liu2015fptas} are `high-temperature' algorithms, exploiting correlation decay properties of the uniform distribution on $\cI(G)$ (or more generally the hard-core model of weighted independent sets) for small vertex degrees.

The second category of bipartite approximation algorithms are those that apply to all bipartite graphs, with running times  better than what is known for general graphs.  The main result in this direction is the algorithm of Goldberg, Lapinskas, and Richerby~\cite{GLR20} which provides an $\eps$-relative approximation to $i(G)$ for  bipartite graphs, running in time $O(2^{.2372 n} (1/\eps)^{O(1)})$, beating the best known running time for general graphs of $O(2^{.268 n} (1/\eps)^{O(1)})$ from the same paper (in comparison, the best known running time for exact counting algorithms for general graphs is $O(2^{.3022n})$~\cite{gaspers2017faster}).

\subsection{Main results}

In this paper we use tools from combinatorics, namely the graph container method of Sapozhenko to give new approximate counting algorithms for independent sets in bipartite graphs.
Our first result is an FTPAS for $i(G)$ for weakly expanding, $d$-regular bipartite graphs, for constant $d$. 

For $z>0$, we say that $\hat z$ is an $\eps$-relative approximation to $z$ if $(1- \eps) \le z/\hat z \le (1+\eps)$. A fully polynomial-time approximation scheme (FPTAS) is an algorithm that for every $\eps>0$ outputs an $\eps$-relative approximation to $i(G)$ and runs in time polynomial in $|V(G)|$ and $1/\eps$.  We let $\mu_G$ denote the uniform distribution on $\cI(G)$.  A polynomial-time sampling scheme for $\mu_G$ runs in time polynomial in $|V(G)|$ and $1/\epsilon$ and outputs an independent set with distribution $\hat \mu$ within $\epsilon$ total variation distance of $\mu_G$. 

 For $\alpha>0$, we say that a $d$-regular bipartite graph $G$ with bipartition $X,Y$ is a \emph{bipartite $\alpha$-expander} if for every $A \subseteq X$ and $A \subseteq Y$ of size at most $|X|/2$, we have
 \begin{align}
\label{eqn:expdefn}
|N(A)|\geq (1+\alpha) |A|.
\end{align}

\begin{theorem}
\label{thmExpander}
There exists a constant $C_1>0$ so that for $d$ fixed and sufficiently large, and $\alpha=\frac{C_1 \log^2 d}{d}$, there is an FPTAS for $i(G)$ and a polynomial-time sampling scheme for $\mu_G$ for the class of $d$-regular, bipartite $\alpha$-expander graphs. 
\end{theorem}

Previous results on expander graphs include an FPTAS for $i(G)$ in the case that $G$ is a typical $d$-regular, random bipartite graph~\cite{jenssen2020algorithms,liao2019counting,chen2021sampling}.  These algorithms exploit the very strong expansion conditions satisfied by a random graph: sets of size $\tilde O(n/d)$ on each side of the bipartition expand by a factor $\tilde \Omega(d)$.  

A natural and powerful generalization of the notion of counting independent sets is to consider weighted independent sets in the form of the hard-core model partition function (also known as the independence polynomial)
$$Z_G(\lam) = \sum_{I \in \cI(G)} \lam^{|I|}\,.$$
The corresponding probability measure on independent sets, known as the hard-core model, is given by 
\[ \mu_{G,\lam} (I ) = \frac{\lam^{|I| }}{Z_G(\lam)} \,.\]
Taking $\lam=1$ gives $i(G)$ and $\mu_G$. In previous works, an FPTAS for $Z_G(\lam)$  for bounded-degree, bipartite $\alpha$-expanders was obtained for $\lam$ much larger than $1$, specifically $\lam \ge K \Delta^{c/\alpha}$ for constants $c, K >1$~\cite{jenssen2020algorithms,chen2021fast,friedrich2020polymer}. In particular, under the expansion conditions of Theorem~\ref{thmExpander}, these algorithms require $\lam \ge \Delta^{\tilde\Omega(\Delta)}$. 

Our next result adapts the algorithm of Theorem~\ref{thmExpander} to work for bipartite $\alpha$-expanders for $\lam$ much smaller than $1$.  
\begin{theorem}
\label{thmLamExpand}
For every $\alpha>0$ there exists a constant $C_2>0$ so that for $d \ge 3$ and $\lam > \frac{C_2 \log d}{d^{1/4}}$ there is an FPTAS for $Z_{G}(\lam)$ and a polynomial-time sampling scheme for $\mu_{G,\lam}$ for the class of $d$-regular, bipartite $\alpha$-expanders.
\end{theorem}
In fact  in proving Theorems~\ref{thmExpander} and~\ref{thmLamExpand} we can interpolate between the two cases, letting the lower bound on $\lam$ shrink as the expansion condition gets stronger. The more general condition obtained is essentially the same as the condition for slow mixing of Glauber dynamics given by Galvin and Tetali~\cite{GT06}.

Our next result is an approximation algorithm for $i(G)$ for all (not necessarily expanding) $d$-regular bipartite graphs $G$, where $d$ may either be constant or growing with the size of the graph, $n$.  When $d \to \infty$ as $n\to \infty$, the algorithm runs in subexponential time.   This algorithm estimates $i(G)$ by separating the contribution from non-expanding sets and expanding sets on each side of the bipartition and uses the expander algorithm of Theorem~\ref{thmExpander} as a subroutine. 
\begin{theorem}
\label{thmMain}
For every $c>0$, there is a randomized algorithm that given a $d$-regular, $n$-vertex bipartite graph $G$ outputs an $n^{-c}$-relative  approximation to $i(G)$ with probability at least $2/3$ and runs in time 
\[  \exp\left(O\left(\frac{n \log^3d}{d}\right)\right)\, .
\] 
\end{theorem}

Note that while the algorithm of Theorem~\ref{thmMain} applies more generally than that of Theorem~\ref{thmExpander}, the algorithmic guarantees are weaker in several senses (in addition to the slower running time): the algorithm uses randomness and the accuracy is limited to being polynomially small in $n$. Moreover we do not provide a corresponding sampling algorithm for Theorem~\ref{thmMain}; the problem is not self-reducible for regular graphs, and so there is not a direct reduction of approximate sampling to counting.  We can overcome this in Theorem~\ref{thmExpander}, however, using the self-reducibility of polymer models.

\subsection{Background}

The study of independent sets  plays a central role in combinatorics since a broad range of problems can be phrased in terms of independent sets in graphs (and more generally hypergraphs). The container method is one of the most  powerful combinatorial tools for studying independent sets. At a high level, the container method exploits a clustering phenomenon exhibited by independent sets which can often be used to deduce useful structural information for typical independent sets in a given graph or hypergraph. For graphs, the method was developed in the early 1980's by Kleitman and Winston~\cite{kleitman1982number, kleitman1980asymptotic} and was independently discovered by Sapozhenko who used the method to enumerate independent sets in regular graphs~\cite{SAPOZHENKO87, sapozhenko2001number}. See  the survey of Samotij~\cite{SAMOTIJ15} for background and examples.  The full potential of the container method only recently became apparent with the powerful generalization of the method to the context of hypergraphs developed by Saxton and Thomason~\cite{saxton2015hypergraph} and Balogh, Morris and Samotij~\cite{balogh2015independent}. These developments have made the container method  one of the most influential tools in modern combinatorics. 

In this paper, we will only need ideas from the theory of graph containers, and our treatment is most closely related to that of Sapozhenko~\cite{SAPOZHENKO87}. A canonical application of Sapozhenko's version of the container method is his proof that the number of independent sets in the $d$-dimensional hypercube, $Q_d$, is asymptotically equal to $2\sqrt{e} \cdot2^{2^{d-1}}$ (a result originally proved by Korshunov and Sapozhenko~\cite{KS83}). See also Galvin's exposition of Sapozhenko's proof~\cite{galvin2019independent}.

Recently, Jenssen and Perkins~\cite{jenssen2020independent} used   Sapozhenko's graph containers for $Q_d$ (and Galvin's extension to weighted independent sets~\cite{galvin2011threshold}) along with the theory of polymer models and the cluster expansion to deduce refined counting estimates and detailed probabilistic information for independent sets in $Q_d$. The polymer models they consider closely resemble those used by Jenssen, Perkins and Keevash~\cite{jenssen2020algorithms} to design approximate counting algorithms in bipartite expander graphs. However, the hypercube $Q_d$ is a far weaker expander than the graphs considered in~\cite{jenssen2020algorithms}, ruling out a direct application of the the cluster expansion method. To overcome this obstacle, they showed that the container method arises naturally as a tool for proving cluster expansion convergence. This synthesis of the cluster expansion and the container method has now seen a number of applications to enumeration problems in combinatorics~\cite{jenssen2020homomorphisms, jenssen2021independent, davies2020upper, balogh2021independent}.

Graph containers have in fact been used previously in a number of results in theoretical computer science, including~\cite{GT06,galvin2007sampling,galvin2007torpid}. These results use Sapozhenko's techniques to prove a type of negative algorithmic result: that the Glauber dynamics (a local Markov chain) for sampling independent sets (or $q$-colorings) in bipartite graphs with sufficient expansion exhibit torpid mixing; that is, the mixing time is exponentially large in the size of the graph. 

 In this paper, we return to the algorithmic context to prove positive results, showing that the container method can be algorithmically implemented to design efficient approximate counting and sampling algorithms for a broad class of bipartite graphs.   There is in fact a connection been the torpid mixing results above and the current algorithmic results: a key step in using polymer models and the cluster expansion as algorithms to approximately count independent sets in bipartite graphs is  showing that most (weighted) independent sets can be accounted for by one of two polymer models that capture small deviations from independent sets that are fully contained in one side of the bipartition.  This step (Lemmas~\ref{lem:polymerapprox} below) amounts to proving a kind of torpid mixing result, closely related to the result of Galvin and Tetali who proved torpid mixing for $d$-regular bipartite $\alpha$-expanders, for $\alpha= \Omega(\log^3 d/d)$~\cite[Corollary 1.3]{GT06}, similar to the class of graphs to which Theorem~\ref{thmExpander} applies.

While Theorem~\ref{thmExpander} gives efficient approximate counting and sampling algorithms for a larger class of bipartite expander graphs than previous approaches, it does not say much about the tractability of \#BIS, beyond ruling out a class of graphs as hard examples.    In another prominent approximation problem of undetermined complexity, the Unique Games problem, efficient approximation algorithms for expander graphs~\cite{arora2008unique,kolla2011spectral,makarychev2010play} were later leveraged via graph decompositions into expanding pieces to find subexponential-time algorithms for all instances~\cite{arora2015subexponential}.  This suggests a very natural goal of finding much faster algorithms for approximating \#BIS.
\begin{question}
Is there a subexponential-time approximation algorithm for \#BIS?
\end{question}
Theorem~\ref{thmMain} makes a small step in this direction, giving subexponential-time algorithms for regular graphs of growing degree, using an expander algorithm as a subroutine to account for the contribution to $i(G)$ from expanding sets.  Non-expanding sets are accounted for separately, also using ideas from graph containers.  It is tempting to think that algorithmic graph decomposition results (e.g.~\cite{gharan2014partitioning,carlson2020efficient}) could be used in conjunction with expander algorithms for \#BIS, but it is not clear to us how to use results that bound the number of edges between expanding pieces  to obtain improved approximation algorithms for \#BIS, and so this remains an interesting direction for future research.

\subsection{Outline}

In Section~\ref{secWarmup} we give a brief warm-up to show how ideas from graph containers can be used for approximately counting independent sets in graphs.  In Section~\ref{secContainers} we present the graph container results we will need for the main algorithmic results.  In Section~\ref{secPolymerModels} we define a polymer model which we will use to approximate the number of independent sets in expander graphs.  In Section~\ref{secExpander} we prove Theorem~\ref{thmExpander}.  In Section~\ref{secGenAlgorithm} we prove Theorem~\ref{thmMain}.  In Section~\ref{secWeighted} we extend the graph container results of Section~\ref{secContainers}  and the polymer model of Section~\ref{secPolymerModels} to the case of weighted independent sets to prove Theorem~\ref{thmLamExpand}.

\section{Warm-up: algorithms from containers}
\label{secWarmup}

In this section, we demonstrate how ideas from graph containers can be used to design faster  algorithms to approximately count independent sets in (not necessarily bipartite) $d$-regular graphs. This section is intended as a warm-up which introduces the interplay between the container method and counting algorithms, and is not needed for the proofs of Theorems~\ref{thmExpander},~\ref{thmLamExpand} and~\ref{thmMain}.

Let us assume that we have an algorithm $A$ that runs in time $a(n)$ on a general (not necessarily $d$-regular) graph $G$ on $n$ vertices and outputs an $\epsilon$-relative approximation to $i(G)$. We will  show that if $G$ is $d$-regular for $d = \omega(1)$, one can obtain an algorithm running in time $a(n/2) \cdot 2^{o(n)}$.

Let $T := n \ln d/d$. We note that throughout the paper, we let $\ln$ denote the natural logarithm and let $\log$ denote $\log_2$. We may then write 

\[
i(G) = i_{< T}(G) + i_{\geq T}(G)
\]
where $i_{<T}(G)$ is the number of independent sets of $G$ of size less than $T$, and  $i_{\geq T}(G)$ is the the number of independent sets of size at least $T$. One can compute $i_{<T}(G)$ by brute-force in time $\binom{n}{T} \cdot \operatorname{poly}(n)$. To compute $i_{\geq T}(G)$, we use a subtle idea of Sapozhenko ~\cite{SAPOZHENKO07} (also see~\cite{IK13},~\cite{KP19}) with a tighter analysis. First, we fix an ordering $\prec$ of the vertices of the graph $G$. Given a subset $S\subseteq V(G)$ and vertex $v\in V(G)$, we let $d_S(v)$ denote the number of neighbor of $v$ in $S$. The following algorithm takes an independent set $I$ of size at least $T$ as input and returns a ``certificate'' $\xi$:

\begin{itemize}
\item $t,i \gets 0$, $\xi \gets (0)^n$, $V_0 \gets V(G)$
\item while $t \leq T$ do
\begin{itemize}
\item $v \gets \operatorname{argmax}_{v \in V_i}d_{V_i}(v)$ with ties broken using $\prec$
\item if $v \in I$
\begin{itemize}
\item $V_{i+1} \gets V_i\setminus (\{v\} \cup N(v))$
\item $t \gets t+1$, $i \gets i+1$, $\xi_i \gets 1$.
\end{itemize}
\item if $v \not\in I$
\begin{itemize}
\item $V_{i+1} \gets V_i\setminus \{v\}$
\item $i \gets i+1$.
\end{itemize}
\end{itemize}
\item return $\xi$
\end{itemize}

It is worth noting that $\xi$ is not the indicator vector of any subset of $V(G)$, but rather an indicator vector describing the steps at which a vertex $v\in I$ is removed from $V_i$. Assume that the algorithm runs for $k = k(\xi)$ steps (i.e., $k$ is that final value that $i$ takes in the execution of the algorithm) when the output is $\xi$ and let $V_\xi=V_k$. A key property of $\xi$ is that it determines both $V_\xi$, and $I \cap (V(G) \setminus V_\xi)$. We may therefore group independent sets according to their certificate $\xi$ and write

\begin{equation}
\label{eqn:containers}
i_{\geq T}(G) = \sum_{\substack{\xi \in \{0,1\}^n, \\|\xi| =T}}i(G[V_\xi]).
\end{equation}

The advantage of this expression is that the number of possible certificates $\binom{n}{T}$ is relatively small and, as we show next, the size of each $V_\xi$ is close to $n/2$. We have therefore exploited the clustering phenomenon of independent sets characteristic of the container method to effectively halve the size of the input graph to our algorithm (albeit at the expense of losing regularity of the graph).

We now give the argument to bound the size of each $V_\xi$. For any subset $S \subseteq V(G)$, let us denote $e_S$ to be the number of edges in the induced subgraph $G[S]$. A straightforward extension of the Hoffman bound (see e.g.~\cite{delsarte1973algebraic,lovasz1979shannon,godsil2008eigenvalue,haemers2021hoffman}), and using the fact that the smallest eigenvalue of a $d$-regular graph is at least $-d$, gives us
\begin{equation*}
2e_S \geq \frac{2d}{n}|S|^2 - d|S|.
\end{equation*}
So for $i \leq k$, we have 

\begin{equation}
\max_{v \in V_i}d_{V_i}(v) \geq \frac{2e_{V_i}}{|V_i|} \geq \frac{d}{n}\left(2|V_i| - n\right),
\end{equation}
and so if $v \in I$, then $|V_{i+1}| \leq |V_i|\left(1 - \frac{2d}{n}\right) + d$, and so at the end of the algorithm, we have $|V_k| \leq \frac{n}{2} + O\left(\frac{n \ln d}{d}\right)$.

Thus, if we have an algorithm $A$ running in time $a(n) = a(n, \epsilon)$ on general graphs on $n$ vertices which outputs an $\epsilon$-relative approximation to $i(G)$ for a general graph $G$ on $n$ vertices, then an $\epsilon$-relative approximation to~\eqref{eqn:containers} may be computed in time $\binom{n}{T} \cdot a\left(\frac{n}{2} + O\left(\frac{n \ln d}{d}\right)\right)$. Combining this with the brute-force computation of $i<T$ gives an algorithm running in time 
\[
\operatorname{poly}(n) \cdot \binom{n}{\frac{n\ln d}{d}}^2 \cdot a\left(\frac{n}{2} + O\left(\frac{n \ln d}{d}\right)\right)
\]
that outputs an $\epsilon$-relative approximation to  $i(G)$ for a $d$-regular graph $G$ on $n$ vertices. So in particular, if $d = \omega(1)$, then using the algorithm of Goldberg, Lapinskas, and Richerby~\cite{GLR20} as a blackbox, we obtain a $2^{(0.134 + o(1)) n}$ time algorithm for approximating $i(G)$ in $d$-regular graphs on $n$ vertices.  We will see below in Section~\ref{secGenAlgorithm} that with more sophisticated ideas this running time can be made subexponential provided that the error $\epsilon=n^{-a}$ for some $a>0$ (Theorem~\ref{thmMain}).

\section{Graph container lemmas}
\label{secContainers}

In this section we introduce results from the theory of graph containers that will be key for the algorithms of Theorems~\ref{thmExpander}, \ref{thmLamExpand} and~\ref{thmMain}. Many of the ideas here have their roots in the aforementioned container method of Sapozhenko~\cite{SAPOZHENKO87}.

We assume throughout that $G$ is a $d$-regular bipartite graph on $2n$ vertices with bipartition $X,Y$, so $|X| = |Y| = n$.  For a subset $A \subseteq X$, we use $W$ to denote $N(A)$. Let us define $[A] := \{u \in X~|~N(u) \subseteq W\}$ to be the \emph{closure} of $A$. Let us call $A$ \emph{$2$-linked} if the subgraph of $G^2$ induced by $A$ is connected.  We say that $A$ is \emph{expanding} if 
$$|W| - |[A]| \geq (C_1/2)\frac{\log^2 d}{d}|W| \,,$$
 where the constant $C_1>0$ will be sufficiently large and chosen later. Otherwise, we say that $A$ is \emph{non-expanding}. 
 
 Let $\mathcal{G}(v,a,w)$ denote the set of $2$-linked expanding sets $A$ such that $A \ni v$, $|[A]| = a$, and $|W| = w$. Let $\mathcal{G}'(v,a)$ be the set of $2$-linked non-expanding sets $A \ni v$ such that $A = [A]$, and $|A| = a$. 

\noindent\textbf{Remark:} Observe that if $G$ is a bipartite $\alpha$-expander (as defined at~\eqref{eqn:expdefn}) then
\begin{equation}
\label{eqn:newexpdefn}
|N(A)| \geq (1+\alpha) |[A]|
\end{equation}
for each $A\subseteq X$ or $A\subseteq Y$ such that $|[A]|\leq n/2$. Indeed, if $|[A]|\leq n/2$, then $|N(A)|=|N([A])|\geq (1+\alpha) |[A]|$.
Since $\alpha \le 1$, inequality~\eqref{eqn:newexpdefn} implies that $|[A]| \leq |N(A)|(1 - \alpha/2)$, and so 
\begin{equation}
\label{eqn:finalexpdefn}
|N(A)| - |[A]| \geq (\alpha/2)|N(A)|
\end{equation}
for each $A$ such that $|[A]|\leq n/2$.
In the following, condition~\eqref{eqn:finalexpdefn} is slightly more convenient to work with and motivates our definition of an expanding set. 

We now state our main technical lemmas.  The first bounds the number of expanding sets and the second bounds the number of non-expanding sets (and gives an algorithm to enumerate them).
\begin{lemma}
\label{lem:exp}
There is an absolute constant $c_1>0$ such that for every $v,a,w$ we have
\[
|\mathcal{G}(v,a,w)| \leq 2^{w - c_1(w - a)}.
\]
\end{lemma}

\begin{lemma}
\label{lem:nonexp}
There is an absolute constant $c_2 >0$ such that for every $v,a$, we have
\[
|\mathcal{G}'(v,a)| \leq 2^{c_2 \frac{a\log^2 d}{d}}.
\]
Moreover, there is an algorithm running in time $2^{O\left(\frac{a\log^2 d}{d}\right)}\cdot \operatorname{poly}(n)$ that outputs the set $\mathcal{G}'(v,a)$.
\end{lemma}

Recall that for a subset $A \subseteq X$, we use $W$ to denote $N(A)$, with $|[A]| = a$, $|W| = w$. Set $t = w - a$. For every $s >0$, let $W_{s} = \{u \in W~|~ d_{[A]}(u) \geq s\}$.   We next define a notion of an approximation of the set $W$ (which in turn determines $[A]$). 
\begin{definition}A set $F \subseteq W$ is an \emph{essential subset} for $A$ if
\begin{enumerate}
\item $F \supseteq W_{d/2}$
\item $N(F) \supseteq [A]$.
\end{enumerate}
\end{definition}
The next lemma gives a family $\mathcal{C}(v,a,w)\subset 2^Y$ that contains an essential subset for each member of $\mathcal{G}(v,a,w)$. Crucially the set of approximating sets $\mathcal{C}(v,a,w)$ is far smaller than $\mathcal{G}(v,a,w)$.

\begin{lemma}
\label{lem:varphi} 
There is a family $\mathcal{C}(v,a,w) \subset 2^Y$ of size at most 
\[
2^{\frac{16 w \log^2 d }{d}}
\] 
such that $\mathcal{C}(v,a,w)$ contains an essential subset of every $2$-linked set $A \ni v$ such that $|[A]| = a$ and $|W| = w$. Moreover, there is an algorithm running in time $2^{\frac{16 w \log^2 d }{d}} \cdot \operatorname{poly}(n)$ that outputs the set $\mathcal{C}(v,a,w)$.
\end{lemma}
We prove Lemma~\ref{lem:varphi} below.

The following lemma of Park~\cite{PARK21}  strengthens a result of Sapozhenko~\cite{SAPOZHENKO87} (the lemma is proved implicitly in~\cite{KP19} also).
\begin{lemma}
\label{lem:container}
There is an absolute constant $c_3 > 0$ such that the following holds: for every $F \subseteq X$, let $\mathcal{G}(F,a,w)$ be the set of expanding $2$-linked sets $A \subseteq X$ such that $|[A]| = a$, $|W| = w$, and $F$ is an essential subset of $A$. Then
\[
|\mathcal{G}(F,a,w)| \leq2^{w - c_3 \left(w-a\right)}.
\] 
\end{lemma}
With these lemmas in hand, we now prove Lemma~\ref{lem:exp} and Lemma~\ref{lem:nonexp}.

\begin{proof}[Proof of Lemma~\ref{lem:exp}]
First note that
\begin{align}\label{eq:union}
    \mathcal{G}(v,a,w)\subseteq\bigcup_{F \in \mathcal{C}(v,a,w)}\mathcal{G}(F,a,w)
\end{align}
and so by Lemmas~\ref{lem:varphi} and~\ref{lem:container}, we have that 
\begin{align*}
|\mathcal{G}(v,a,w)| & \leq \sum_{F \in \mathcal{C}(v,a,w)}|\mathcal{G}(F,a,w)| \\
& \leq |\mathcal{C}(v,a,w)| \cdot \max_{F \in \mathcal{C}(v,a,w)}|\mathcal{G}(F,a,w)| \\
& \leq 2^{\frac{16w \log^2 d}{d}} \cdot 2^{w - c_3\left(w-a\right)} \\
& \leq 2^{w - \left(\frac{c_3}{2}\right)\left(w-a\right)} 
\end{align*}
where for the last inequality we used that $w-a\geq (C_1/2)\frac{\log^2 d}{d}w$ by the definition of an expanding set and assumed that $C_1 > 64/c_3$. 
\end{proof}

\begin{proof}[Proof of Lemma~\ref{lem:nonexp}]
Let $F$ be an essential subset of $A$ where $A$ is non-expanding.
Note that each vertex in $W \setminus F$ has at least $d/2$ neighbors in $[A]^c$, and there are at most $d(w-a)$ edges between $W$ and $[A]^c$. It follows that
\[
(d/2) \cdot |W \setminus F| \leq  d(w-a)\
\]
 and so
\[
|W \setminus F| \leq  2(w-a) \leq C_1 \frac{\log^2 d}{d}w\, .
\]

Moreover, $W \setminus F \subset N^2(F)$, and $N^2(F) \leq wd^2$, and so there are at most \[\binom{wd^2}{ \leq C_1w(\log^2 d)/d} \leq 2^{O\left(\frac{a\log^3 d}{d}\right)}
\]
choices for $W$, each of which determines a $[A]$. Let 
$\mathcal{G}'(F,a)$ denote the collection of $A\subset X$ such that $A=[A]$, $v\in A$, $A$ is $2$-linked, non-expanding and $F$ is an essential subset for $A$. Then by the above
\[
|\mathcal{G}'(F, a)|= 2^{O\left(\frac{a\log^3 d}{d}\right)}\, .
\]
Moreover, we can generate the set $\mathcal{G}'(F, a)$ in time $2^{O\left(\frac{a\log^3 d}{d}\right)}\operatorname{poly}(n)$ by listing each set $W$ that is a union of $F$ with a subset of $N^2(F)$ of size at most $C_1w(\log^2 d)/d$, generating the corresponding closed set $[A]$ such that $N(A)=W$, and checking it satisfies the required conditions.

Now, by Lemma~\ref{lem:varphi},
\begin{align}\label{eq:union2}
\mathcal{G}'(v, a)\subseteq \bigcup_{w \in \left[a,a\left(1 + C_1\frac{\log^2 d}{d}\right)\right]}\bigcup_{F \in \mathcal{C}(v,a,w)}\mathcal{G}'(F,a)
\end{align}
and so 
\[
|\mathcal{G}'(v, a)|\leq \left(C_1\frac{a\log^2 d}{d} \right) \cdot  2^{O\left(\frac{a\log^3 d}{d}\right)} = 2^{O\left(\frac{a\log^3 d}{d}\right)}\, .
\]
Similarly, by~\eqref{eq:union2}, Lemma~\ref{lem:varphi}, and the above algorithm for generating $\mathcal{G}'(F,a)$, we may generate $\mathcal{G'}(v,a)$ in time $2^{O\left(\frac{a\log^3 d}{d}\right)}\cdot \operatorname{poly}(n)$.
\end{proof}

We will need a covering result originally due to Lov\'{a}sz~\cite{LOVASZ75} and Stein~\cite{STEIN74}.
\begin{theorem}
\label{thm:cover}
Let $H$ be a bipartite graph on vertex sets $P$ and $Q$ where the degree of each vertex in $P$ is at least $a$ and the degree of each vertex in $Q$ is at most $b$. Then there is subset $Q' \subset Q$ of size at most $\frac{|Q|}{a}(1 + \ln b)$ such that $P \subseteq N(Q')$.
\end{theorem}

We record the following corollary of Theorem~\ref{thm:cover} which we will use in the proof of Lemma~\ref{lem:varphi}.
\begin{corollary}\label{cor:cover}
Let $A\subseteq X$ be $2$-linked and let $A\ni v$. Then the following hold:
\begin{enumerate}
\item There exists a $2$-linked subset $A'\subseteq A$ of size at most $2\frac{a}{d} \ln d+ 2\frac{w}{d}$ such that $A'\ni v$ and $N(A')$ is an essential subset for $A$. \label{item1}
\item There exists a $2$-linked subset $A''\subseteq A$ of size at most $2\frac{a}{d} \ln d+ 2\frac{w}{d}+ 2(w-a)$ such that $N(A'')=W$. \label{item2}
\end{enumerate}

\end{corollary}
\begin{proof} We begin by proving~\eqref{item1}.
Let $A_0 \subset [A]$ be a maximal subset of vertices containing $v$ with pairwise disjoint neighborhoods. Clearly $|A_0| \leq \frac{w}{d}$ and $N^2(A_0) \supseteq A$. Theorem~\ref{thm:cover} guarantees a subset $A_1 \subseteq A$ of size at most $2\frac{a}{d}\ln d$ such that $W_{d/2} \supseteq N(A_1)$. Suppose $A_0 \cup A_1$ is not $2$-linked, then there are at most $\frac{w}{d}$ $2$-linked components. Indeed, this is true since $N(A_0 \cup A_1)\subseteq W$ and each two linked component covers at least $d$ vertices of $W$. Since $[A]$ is $2$-linked, it follows that one can choose a subset $A_2 \subseteq [A]$ of size at most $\frac{w}{d}$ such that $A':=A_0 \cup A_1 \cup A_2$ is $2$-linked. We note that $|A'|\leq 2\frac{a}{d} \ln d+ 2\frac{w}{d}$. To show that $N(A')$ is an essential subset for $A$
observe that 
\begin{itemize}
\item $N^2(A') \supseteq N^2(A_0) \supseteq [A]$, and 
\item $W_{d/2} \subseteq N(A_1)  \subseteq N(A') $.
\end{itemize}

We now turn to~\eqref{item2}.
Note that each vertex in $W \setminus W_{d/2}$ has at least $d/2$ neighbors in $[A]^c$, and there are at most $d(w-a)$ edges between $W$ and $[A]^c$. It follows that
\[
(d/2) \cdot |W \setminus W_{d/2}| \leq  d(w-a)\
\]
 and so
\[
|W \setminus W_{d/2}| \leq  2(w-a).
\]
Let $A_3 \subseteq A$ be a minimal cover of $W \setminus W_{d/2}$. We have that $|A_3| \leq |W \setminus W_{d/2}| \leq 2(w-a)$, and every vertex of $A_3$ is at a distance $2$ from some vertex in $A_0\subseteq A'$ by the maximality of $A_0$. Thus $A''=A'\cup A_3$ is $2$-linked, $|A''|\leq 2(w-a)+|A'|$ and $N(A'')\supseteq W$,
completing the proof.
\end{proof}

Finally we prove Lemma~\ref{lem:varphi}.
The proof  we present is different and simpler than the one in~\cite{SAPOZHENKO87}, whose proof works for a quantitatively weaker notion of expansion, but in return, asks that no two vertices share many common neighbors.
\begin{proof}[Proof of Lemma~\ref{lem:varphi}]
Let $A\ni v$ be a $2$-linked subset as in the statement of the lemma. By Corollary~\ref{cor:cover}, there exists a $2$-linked subset $A'\subset A$ of size at most $4\frac{w}{d} \ln d$ such that $v\ni A'$ and $N(A')$ is an essential subset for $A$.
In view of this, we let $\mathcal{B}(v,a,w)$ be the set of all $2$-linked subsets of $A$, containing $v$, of size at most $4\frac{w}{d} \ln d$ and let 
\[
\mathcal{C}(v,a,w)=\left\{N(A): A\in \mathcal{B}(v,a,w)\right\}\, .
\]
It remains to upper bound the size of $\mathcal{C}(v,a,w)$ and describe an algorithm that outputs the set. Note that $|\mathcal{B}(v,a,w)|=|\mathcal{C}(v,a,w)|$ and $|\mathcal{B}(v,a,w)|$ is at most the number of trees in $G^2$ containing $v$ as the root with at most $4 \frac{w \log d}{d}$ vertices. Note that the maximum degree in $G^2$ is at most $d(d-1)$. So $\mathcal{B}(v,a,w)$ can be enumerated using the following procedure:
\begin{enumerate}
\item Assume an ordering on the neighbors of every vertex in the graph $G^2$. Let us use $v_i$ to denote the $i$'th neighbor of a vertex $v$. Let us also denote $v_0 = v$.
\item Generate a list $S \in\{0,\ldots, d(d-1)\}^{8 \frac{w \ln d}{d}}$.
\item Consider the set $T_S = \left\{v^{(0)},\ldots,v^{(s)}\right\}$ where $v^{(0)} = v$ and $v^{(i)} = v^{(i-1)}_{S_i}$.
\item If $|T_S| \leq 4\frac{w \ln d}{d}$, then output $T_S$.
\end{enumerate}
Consider any tree in $G^2$ with root $v$ and $s \leq 4\frac{w \ln d}{d}$ nodes. There is at least one choice of the list $S$ that causes the above procedure to output the vertices of this tree, namely, if $(S_1,\ldots,S_{2s})$ is the DFS traverse order and $S_i = 0$ for $i \geq 2s$.

For each list $S$, the procedure takes $\operatorname{poly}(n)$ time, and the number of possible lists is \linebreak $(d(d-1) + 1)^{8\frac{w}{d}\ln d} \leq 2^{\frac{16w \log^2 d}{d}}$. Therefore there is an algorithm running in time \linebreak $2^{\frac{16w \log^2 d}{d}} \cdot \operatorname{poly}(n)$ that outputs the set $\mathcal{C}(v,a,w)$.
\end{proof}

\section{Polymer Models}
\label{secPolymerModels}

In this section we introduce a variant of the polymer models used by Jenssen, Keevash, and Perkins~\cite{jenssen2020algorithms} to obtain approximation algorithms for bipartite expander graphs.  Polymer models originated in statistical physics (e.g.~\cite{gruber1971general,kotecky1986cluster}) as a means to study spin models on lattices.  Recently they were used to design algorithms for spin models at low temperatures~\cite{helmuth2020algorithmic}.

Fix a $d$-regular bipartite graph $G$ with bipartition $(X,Y)$ of size $n$ each.  A \emph{polymer} of $G$ is a $2$-linked, expanding subset of $X$. Recall, from the previous section, that a set $A \subseteq X$ is expanding if 
\begin{align}\label{def:expanding}
|N(A)| - |[A]| \geq (C_1/2)\frac{|N(A)| \log^2 d}{d}
\end{align}
where $C_1$ is the constant from Theorem~\ref{thmExpander}. Let $\mathcal{P}(G)$ denote the set of all polymers of $G$. The \emph{weight} of a polymer $\gamma$ is given by $2^{-|N(\gamma)|} =: w_{\gamma}$. We call two polymers $\gamma$ and $\gamma'$ \emph{compatible} if $\gamma \cup \gamma'$ is not $2$-linked, and \emph{incompatible} otherwise, in which case we write $\gamma \nsim \gamma'$. Let $\Omega(G)$ denote the collection all subsets of mutually compatible polymers (including the empty set). The polymer model partition function is
\begin{equation}
\label{eqn:PPF}
\Xi_{G}^X := \sum_{\Lambda \in \Omega(G)} \prod_{\gamma \in \Lambda} w_{\gamma}\, ,
\end{equation}
and the associated Gibbs distribution $\nu^X_G$ on $\Omega(G)$ defined by
\[
\nu^X_G(\Lambda)=\frac{\prod_{\gamma \in \Lambda}w_{\gamma}}{\Xi_{G}^X } \text{ for } \Lambda\in \Omega(G) \,.
\]

When the weights of the polymer model are small enough one can hope to understand the partition function  and the Gibbs distribution via perturbative techniques, and in particular, the cluster expansion.

For a tuple $\Gamma$ of polymers, the \textit{incompatibility graph}, $H(\Gamma)$, is the graph with vertex set $\Gamma$ and an edge between any two incompatible polymers.  A \textit{cluster} $\Gamma$ is an ordered tuple of polymers so that $H(\Gamma)$ is connected.
Let us use $\mathcal{C}(G)$ to denote  the set of all clusters of $G$. The \emph{cluster expansion} of $\Xi_G^X$ is the formal series expansion

\begin{equation}
\label{eqn:CE}
\ln \Xi_G^X = \sum_{\Gamma \in \mathcal{C}(G)} \phi(H(\Gamma))\prod_{\gamma \in \Gamma} w_{\gamma}\, ,
\end{equation}

where $\phi$ is the Ursell function defined by
\[
\phi(H) =  \sum_{\substack{E' \subseteq E(H) \\ (V(H),E') \text{ connected}  }} (-1)^{|E'|} \,.
\]
Since $\mathcal{C}(G)$ is an infinite set, the series in~\eqref{eqn:CE} is an infinite sum. In our application, we will work with a truncated sum after after establishing a fast enough rate of  convergence. For these convergence bounds, (as in~\cite{helmuth2020algorithmic,jenssen2020algorithms}) we use a special case of the Koteck{\`y}--Preiss condition~\cite{kotecky1986cluster}.
\begin{theorem}[\cite{kotecky1986cluster}]
\label{thm:KPcondition}
Fix functions $f: \mathcal{P}(G) \rightarrow [0,\infty)$ and $g: \mathcal{P}(G) \rightarrow [0,\infty)$. Suppose that for every $\gamma \in \mathcal{P}(G)$, we have
\begin{equation}
\label{eqn:KPcondition}
\sum_{\gamma' \not\sim \gamma}w_{\gamma'}e^{f(\gamma') + g(\gamma')} \leq f(\gamma).
\end{equation}
Then the cluster expansion converges absolutely. Moreover, for every vertex $v$,
\begin{equation}
\sum_{\substack{\Gamma \in \mathcal{C}(G)\\ \Gamma \ni v}}\left|\phi(\Gamma)\prod_{\gamma \in \Gamma}w_{\gamma}\right| e^{g(\Gamma)} \leq 1.
\end{equation}
\end{theorem}

We will use this to obtain some relevant structural information about the polymer model. For  polymers $\gamma$, define modified polymer weights $\tilde{w}_{\gamma} = 2^{-|N(\gamma)|} \cdot 2^{  \frac{|\gamma|  \log^2 d}{d}}$ and let 
\begin{align}\label{eq:xitilde}
\tilde{\Xi}_{G}^X = \sum_{\Lambda \in \Omega(G)}\prod_{\gamma \in \Lambda}\tilde{w}_{\gamma}
\end{align}
be the modified polymer model partition function. We first show that this polymer model satisfies the Koteck{\`y}-Preiss condition for a suitable choice of functions $f$ and $g$.
\begin{lemma}
\label{lem:expKP}
Consider the modified polymer model where the polymers are all the $2$-linked subsets $A \subseteq X$ that are expanding, and the weight of a polymer $\gamma$ is given by
\[
\tilde{w}_{\gamma} = 2^{-|N(\gamma)|} \cdot 2^{\frac{|\gamma| \log^2 d}{d}}.
\]Let $f(\gamma)  =  \ln 2 \cdot \frac{|\gamma|\log^2 d}{d}$, and $g(\gamma) = 2 \ln 2 \cdot \frac{|N(\gamma)|\log^2 d}{d}$. Then for $d$ sufficiently large,  every polymer $\gamma$ satisfies~(\ref{eqn:KPcondition}).  The conclusion  also holds  for the original polymer model with weights $w_{\gamma}$ and the same choice of functions $f(\cdot), g(\cdot)$.
\end{lemma}

\begin{proof}[Proof of Lemma~\ref{lem:expKP}]
It suffices to prove the claim for the modified polymer model since  $\tilde{w}_{\gamma}>w_{\gamma}$ for all $\gamma$.

Recall, from the proof of Lemma~\ref{lem:exp} that $c_1 \geq c_3/2 \geq 64/C_1$. We evaluate
\begin{align*}
\sum_{\gamma' \not\sim \gamma}\tilde{w}_{\gamma'}e^{f(\gamma') + g(\gamma')} & \leq \sum_{v \in \gamma} \sum_{\gamma' \not \sim v}\tilde{w}_{\gamma'}e^{f(\gamma') + g(\gamma')}\\
& \leq |\gamma| \cdot \max_{v\in \gamma} \sum_{\gamma' \not\sim v}2^{-|N(\gamma')|}e^{2f(\gamma') + g(\gamma')}\\
& \leq |\gamma| \cdot \max_{v\in \gamma} \sum_{w \geq d} \left(\sum_{t \geq (C_1/2)\frac{w \log^2 d}{d}} |\mathcal{G}(v,w-t,w)| 2^{-w}\cdot e^{4 \ln 2\cdot \frac{w\log^2 d}{d} } \right)\\
& \leq |\gamma| \cdot \max_{v\in \gamma} \sum_{w \geq d}e^{4 \ln 2\cdot \frac{w\log^2 d}{d}}  \left(\sum_{t \geq (C_1/2)\frac{w \log^2 d}{d}} |\mathcal{G}(v,w-t,w)| 2^{-w } \right)\\
& \leq |\gamma| \sum_{w \geq d}e^{4 \ln 2\cdot \frac{w\log^2 d}{d}}  \left(\sum_{t \geq (C_1/2)\frac{w \log^2 d}{d}} 2^{-c_1 \cdot t} \right)\\
& \leq d|\gamma| \sum_{w \geq d}e^{4 \ln 2\cdot \frac{w\log^2 d}{d} }  2^{-8\frac{w\log^2 d}{d}} \\
& \leq d^2|\gamma|2^{-4\log^2 d}\\
& \ll \ln 2\cdot \frac{|\gamma| \log^2 d}{d} = f(\gamma) \, . \qedhere
\end{align*}
\end{proof}

Therefore, Theorem~\ref{thm:KPcondition} gives us that 
\begin{align}
\label{eqn:tailvertex}
\sum_{\substack{\Gamma \in \mathcal{C}(G)\\ \Gamma \ni v }}\left|\phi(\Gamma)\prod_{\gamma \in \Gamma}w_{\gamma}\right| e^{g(\gamma)}\leq 1 \,,
\end{align}
and the same with $\tilde{w}_{\gamma}$ replacing $w_{\gamma}$.

Now define the exponential of the truncated cluster expansion
\begin{equation}
\label{eqn:xihat}
{\Xi}_G^X(\ell) :=  \exp\left(\sum_{\substack{\Gamma \in \mathcal{C}(G) \\ \|\Gamma\| \leq \ell}} \phi(\Gamma)\prod_{\gamma \in \Gamma} w_{\gamma}\right),
\end{equation}
where $\|\Gamma\|:=\sum_{\gamma\in \Gamma}|\gamma|$.
The following lemma bounds the error in approximating  $\Xi_G^X$  by  ${\Xi}_G^X(\ell)$.

\begin{lemma}
\label{lem:xihat}
We have for every $\ell\geq 1$,
\begin{equation}
\label{eqn:APXCE}
\left|\ln\Xi_G^X-  \ln {\Xi}_G^X(\ell) \right| \le  n \cdot 2^{-2\frac{\ell \log^2 d}{d}}\, .
\end{equation}
In particular, if $\ell\geq \frac{d}{2\log^2 d} \log(n/\eps)$, then 
\[
\left|\ln\Xi_G^X-  \ln {\Xi}_G^X(\ell) \right| \le \eps \, .
\]
\end{lemma}
\begin{proof}
First recall that for a cluster $\Gamma$, 
\[g(\Gamma)=\sum_{\gamma\in \Gamma}g(\gamma)=2 \ln 2\frac{\log^2 d}{d}\sum_{\gamma\in \Gamma}|N(\gamma)| \geq 2 \ln 2\frac{\log^2 d}{d}\|\Gamma\|.
\]
It follows from \eqref{eqn:tailvertex} that 

\begin{align*}
\left|\ln\Xi_G^X-  \ln {\Xi}_G^X(\ell) \right|
&=
 \left|\sum_{\substack{\Gamma \in \mathcal{C}(G) \\  \|\Gamma\|>\ell}}\phi(\Gamma)\prod_{\gamma \in \Gamma}w_{\gamma}\right|\\
 &\leq 
 \sum_v \sum_{\substack{\Gamma \in \mathcal{C}(G)\\ \Gamma \ni v\\ \|\Gamma\| > 
 \ell}}\left|\phi(\Gamma)\prod_{\gamma \in \Gamma}w_{\gamma}\right|\\
 &\leq
 ne^{-2 \ln 2\frac{\ell\log^2 d}{d}} \\
 &= n \cdot 2^{-2 \frac{\ell \log^2 d}{d}}\, . \qedhere
\end{align*}
\end{proof}

Let $\|\mathbf{\Lambda}\| = \sum_{\gamma 
\in \mathbf{\Lambda}}|\gamma|$ for $\mathbf{\Lambda} \sim \nu_G^X$. We have the following large deviation result for $\|\mathbf{\Lambda}\|$, following~\cite[Lemma 16]{jenssen2020independent}.
\begin{lemma}
\label{lem:tailbound}
For any $\delta \in (0,1)$, there is a $d_0 = d_0(\delta)$ such for $d \geq d_0$, we have
\[
\mathbb{P}(\|\mathbf{\Lambda}\| \geq \delta n) \leq  2^{-\left(\frac{\delta n \log^2 d}{2d}\right)}.
\]
\end{lemma}
\begin{proof}
With $\tilde{\Xi}_G^X$ as defined at \eqref{eq:xitilde}, we have that $\ln \tilde{\Xi}_G^X - \ln \Xi_G^X = \ln \mathbf{E}e^{\zeta \cdot\|\mathbf{\Lambda}\|}$ for $\zeta = \frac{ \ln 2 \log^2 d}{d}$.  Summing~(\ref{eqn:tailvertex}) over all $v$, and using the fact that $|N(\gamma)| \geq d$ for every $\gamma$, we get
\begin{equation}
\label{eqn:pfbound}
\ln \tilde{\Xi}_G^X \leq \sum_v\sum_{\substack{\Gamma \in \mathcal{C}(G)\\ \Gamma \ni v }}\left|\phi(\Gamma)\prod_{\gamma \in \Gamma}\tilde{w}_{\gamma}\right| \leq n \cdot 2^{-2 \log^2 d}.
\end{equation}
Therefore, we have 
\begin{align*}
\ln \mathbf{E}e^{\zeta \cdot \|\mathbf{\Lambda}\|} \leq \ln \tilde{\Xi}_G^X \leq n \cdot 2^{-2 \log^2 d} \,,
\end{align*}
and so by Markov's inequality we have
\begin{align*}
\mathbb{P}(\|\mathbf{\Lambda}\| \geq \delta n)  & \leq \exp\left(- \zeta\delta n + n\cdot 2^{-2 \log^2 d}\right) \\
&\leq 2^{- \frac{\delta n \log^2 d}{2d}},
\end{align*}
where the last inequality follows because $ (1/2)\ln 2 \geq   (\delta \log^2 d)^{-1} \cdot d \cdot 2^{-2 \log^2 d}$ for  large enough $d$. 
\end{proof}

\subsection*{Approximate counting and sampling for polymer models}
\label{sec:CEalgorithm}
Here, we will use an algorithm from~\cite{jenssen2020algorithms} to approximate $\Xi_G^X$ and approximately sample from $\nu_{G}^X$. 

\begin{lemma}
\label{lem:CEalgorithm}
Then there is an FPTAS for $\Xi_{G}^X$ that runs in time $(n/\eps)^{O(d)}$. Moreover, for every $\eps>0$ there is a randomized algorithm that runs in time
$(n/\eps)^{O(d)}$ that outputs a configuration  $\Lambda \in \Omega(G)$ with distribution $\nu^X_{\text{alg}}$
so that $\|\nu^X_{\text{alg}}-\nu_{G}^X\|_{TV}<\eps$. 
\end{lemma}
To give a sense of the above algorithms, we briefly describe the FPTAS for $\Xi_{G}^X$. Using Lemma~\ref{lem:xihat}, it is enough to compute $\Xi_{G}^X(\ell)$ (as defined in~\eqref{eqn:xihat}) where $\ell=\left\lceil \frac{d}{2\log^2 d} \log(n/\eps) \right\rceil$.

This may be done by 

\begin{enumerate}
\item listing all clusters $\Gamma$ of size at most $\ell$,
\item computing $\phi(\Gamma)$,
\item computing $\prod_{\gamma \in \Gamma}w_{\gamma}$, and 
\item evaluating $\Xi_{G}^X(\ell)$ by exponentiating the truncated cluster expansion.
\end{enumerate}

To approximately sample from $\nu_{G}^X$ we appeal to the self-reducibility of the abstract polymer model, see~\cite[Theorem 10]{helmuth2020algorithmic}.

\begin{remark}\label{rem:subg} In the proof of Theorem~\ref{thmMain} in Section~\ref{secGenAlgorithm}, we will in fact need to approximate the partition function of various subgraphs of $G$. These subgraphs, $H\subseteq G$, all have the following form: $H$ is the subgraph induced on vertex set $X'\cup Y'$ for some $X'\subseteq X$, $Y'\subseteq Y$ such that $N(X')\subseteq Y'$. For such a subgraph, we define a polymer of $H$ to be an expanding $2$-linked subset of $X'$ and define the partition function $\Xi_H^{X'}$ in the obvious way. We note that the polymers of $H$ are a subset of the polymers of $G$. In particular, since the polymer model on $G$ satisfies the hypothesis of Theorem~\ref{thm:KPcondition}, the polymer model on $H$ also satisfies the hypothesis of Theorem~\ref{thm:KPcondition} with the same functions $f$ and $g$. In particular, Lemmas~\ref{lem:xihat},~\ref{lem:tailbound} and~\ref{lem:CEalgorithm} also apply to the polymer model on $H$. 
\end{remark}

\section{An algorithm for expander graphs: proof of Theorem~\ref{thmExpander}}
\label{secExpander}
In this section, we prove Theorem~\ref{thmExpander}. We assume throughout that  $G$ is a $d$-regular bipartite $\alpha$-expander with $\alpha = C_1 \frac{\log d}{d}$. We let $X, Y$ denote the vertex classes of $G$ and let $n=|X|=|Y|$. We note that by~\eqref{eqn:finalexpdefn}, we have that $|N(A)| - |A|| \geq (C_1/2) \frac{|N(A)|\log^2 d}{d}$ for every $A \subset X$ or $A \subset Y$ such that $|[A]| \leq n/2$, i.e. each such set $A$ is expanding.

First we show that  $i(G)$ can be  approximated well by a linear combination of the polymer model partition functions $\Xi_G^X$ and $\Xi_G^Y$ (as defined in~\eqref{eqn:PPF}). We may then use the algorithm of Section~\ref{sec:CEalgorithm} to approximate $\Xi_G^X$ and $\Xi_G^Y$. 
Recall that we let $\mathcal I = \cI(G)$ denote the set of independent sets of $G$.
For the sampling algorithms, we show that $\mu_G$ can be approximated by a mixture of probability distributions on $\mathcal I$ derived from the polymer measures $\nu^X_G$, $\nu_G^Y$.
We let $\hat\nu^X_G$ denote the probability distribution on $\mathcal I$ defined as follows:
\begin{enumerate}
\item Sample a collection of compatible polymers $\Lambda$ from the measure $\nu^X_G$.
\item Set $I=J \cup \bigcup_{\gamma\in\Lambda}\gamma$ where $J$ is a uniformly random subset of $Y\backslash \bigcup_{\gamma\in\Lambda}N(\gamma)$.
\end{enumerate} 
We define $\hat\nu^Y_G$ analogously and define the mixture
\[
\hat \mu_G= \frac{\Xi_G^X}{\Xi_G^X+\Xi_G^Y} \hat\nu^X_G + \frac{\Xi_G^Y}{\Xi_G^X+\Xi_G^Y} \hat\nu^Y_G\, .
\]

\begin{lemma}
\label{lem:polymerapprox}
For $n$ sufficiently large,
\begin{equation*}
2^n \cdot \left(\Xi_G^X + \Xi_G^Y\right)
\end{equation*}
is an $\eps$-relative approximation to $i(G)$ where
$\eps= 2^{-\frac{ n \log^2 d}{60d}}.$
Moreover, 
\[
\|\mu_G - \hat \mu_G\|_{TV}\leq 2\eps\, .
\]
\end{lemma}
\begin{proof}
Let us define
\[
\mathcal{I}_X := \{I \in \mathcal{I}~|~\text{every $2$-linked component of $I \cap X$ is expanding}\}\,
\]
i.e.\ the set of all $I$ such that $\nu_G^X(I)>0$. We note that since the independence number of $G$ is $n$, we have that for each $I \in \mathcal{I}$, 
  \[ \min(|[I \cap X]|, |[I \cap Y]|) \leq n/2 \,. \]
   In particular, every component of either $I\cap X$ or $I \cap Y$ is expanding. Therefore $\mathcal{I} = \mathcal{I}_X \cup \mathcal{I}_Y$ and so
   \begin{align}\label{eqiGinc}
   i(G)=|\mathcal{I}_X|+ |\mathcal{I}_Y|-|\mathcal{I}_X \cap \mathcal{I}_Y|.
   \end{align}
Note that
\begin{equation}
\label{eqn:defn}
\mathcal{I}_X = 2^n \cdot \Xi_G^X
\end{equation}
and moreover, the set of $2$-linked components of $I \cap X$ for a uniformly chosen $I \in \mathcal{I}_X$ is distributed exactly accordingly to $\nu_G^X$. It will suffice to bound the size of $\mathcal{I}_X \cap \mathcal{I}_Y$.

Letting $\mathcal{I}^\delta_X=\{I\in \mathcal I_X: |I\cap X|\leq \delta n\}$, for $\delta\in (0,1)$, it follows by Lemma~\ref{lem:tailbound} that
\begin{equation}
\label{eqn:Xlarge}
|  \mathcal{I}_X\backslash\mathcal{I}^\delta_X | \leq |\mathcal{I}_X| \cdot 2^{- \frac{\delta n \log^2 d}{2d}}.
\end{equation}
Defining $\mathcal{I}_Y$, $\mathcal{I}^\delta_Y$, analogously and taking $\delta=1/30$ we have 
\begin{equation}
\label{eqn:Ysmall}
|\mathcal{I}^\delta_X \cap \mathcal{I}^\delta_Y| \leq \binom{2n}{\leq 2\delta n}\leq 2^{n}\cdot 2^{-\frac{ n \log^2 d}{60d}- 1} \leq |\mathcal{I}_X|\cdot 2^{-\frac{n \log^2 d}{60d} - 1}.
\end{equation}
By~\eqref{eqn:defn},~\eqref{eqn:Xlarge} and \eqref{eqn:Ysmall} we conclude that
    \begin{align}\label{eqintbound}
    |\mathcal{I}_X \cap \mathcal{I}_Y|\leq |\mathcal{I}_X\backslash\mathcal{I}_X^\delta| + |\mathcal{I}_Y\backslash\mathcal{I}_Y^\delta| + |\mathcal{I}^\delta_Y\cap\mathcal{I}_X^\delta| \leq 2^n(\Xi_G^X + \Xi_G^Y)2^{\frac{- n \log^2 d}{60d}}\, ,
    \end{align}
    and therefore, by~\eqref{eqiGinc},
\begin{align}\label{eqiapprox}
2^n(\Xi_G^X + \Xi_G^Y)\cdot \left(1 - 2^{\frac{- n \log^2 d}{60d}}\right) \leq i(G) \leq 2^n(\Xi_G^X + \Xi_G^Y).
\end{align}
This completes the proof of the first claim. For the second claim we recall the following formula for the total variation distance between discrete probability measures:
\begin{align}\label{eqTVform}
\|\mu_G-\hat\mu_G\|_{TV}= \sum_{I: \hat\mu_G(I) > \mu_G(I)}\hat\mu_G(I) - \mu_G(I)\, .
\end{align}
We note that for $I\in  \mathcal{I}_X \triangle \mathcal{I}_Y$, 
$\hat\mu_G(I)=2^{-n}(\Xi_G^X + \Xi_G^Y)^{-1}$, whereas for  $I\in  \mathcal{I}_X \cap \mathcal{I}_Y$, 
$\hat\mu_G(I)=2^{1-n}(\Xi_G^X + \Xi_G^Y)^{-1}$. It follows from~\eqref{eqiapprox} that $\hat\mu_G(I) > \mu_G(I)$ only if $I\in  \mathcal{I}_X \cap \mathcal{I}_Y$. By~\eqref{eqintbound} and~\eqref{eqTVform}, we then have
\[
\|\mu_G-\hat\mu_G\|_{TV}\leq 2\cdot 2^{\frac{- n \log^2 d}{60d}}\, . \qedhere
\]
\end{proof}

Now we  prove Theorem~\ref{thmExpander}.
\begin{proof}[Proof of Theorem~\ref{thmExpander}]
Set $\eps_0= 2^{- \frac{ n \log^2 d}{60d}}$.
First suppose $\epsilon \leq 2\eps_0$, then $i(G)$ may be computed exactly and a uniformly random independent set can be sampled by brute-force in time $2^{n + o(n)} = (1/\epsilon)^{O(d/\log^2 d)}$.

Now suppose $\epsilon> 2\eps_0$. By Lemma~\ref{lem:polymerapprox},  it is enough to compute an $(\epsilon/4)$-relative approximation to both $\Xi_G^X$ and $\Xi_G^Y$. By using the algorithm given by Lemma~\ref{lem:CEalgorithm}, this takes time $(n/\epsilon)^{O(d)}$.  

For the approximate  sampling algorithm, we note that by Lemma~\ref{lem:polymerapprox},
it is enough to obtain an $\eps/2$-approximate sample from $\hat\mu_G$. We do this as follows:
we first compute $\eps/8$-relative approximations to $\Xi_G^X$ and $\Xi_G^Y$ by computing $\Xi_G^X(\ell)$ and $\Xi_G^Y(\ell)$ respectively, with $\ell$ chosen as in Lemma~\ref{lem:CEalgorithm}.    We then pick $X$ or $Y$ with respective probabilities $\Xi_G^X(\ell) /(\Xi_G^X(\ell)+\Xi_G^Y(\ell))$ and $\Xi_G^Y(\ell) /(\Xi_G^X(\ell)+\Xi_G^Y(\ell))$, and  then use the polymer sampling algorithm of Lemma~\ref{lem:CEalgorithm} to approximately sample a configuration of compatible polymers $\Lambda$ from $X$ (resp. $Y$), accurate to within total variation distance $\eps/8$.  Given the polymer configuration $\Lambda$ we then independently select each vertex of $Y \setminus N(\Lambda)$ (resp. $X \setminus N(\Lambda)$) with probability $1/2$ and add these to the independent set.  The distribution of the output is then within total variation distance $\eps/2$ of $\hat\mu_G$.   See the sampling algorithm of~\cite{jenssen2020algorithms} for details on the calculation of this bound.
\end{proof}

\section{An algorithm for general regular bipartite graphs: Proof of Theorem~\ref{thmMain}}
\label{secGenAlgorithm}

In this section we prove Theorem~\ref{thmMain}, giving an algorithm that, for any constant $c>0$, returns an $n^{-c}$-relative approximation for the number of independent sets in a general $d$-regular bipartite graph. As in the previous sections we let $G$ denote a $d$-regular graph on $2n$ vertices with vertex classes $X$ and $Y$. The algorithm proceeds by separating the contribution of expanding and non-expanding $2$-linked sets of $X$ to the independent set count. To estimate the the contribution from non-expanding components, we use a simple argument inspired by the container method (see Lemma~\ref{lem:countcover} below). To estimate the contribution from expanding components, we appeal to the algorithm of Lemma~\ref{lem:CEalgorithm}. 

We begin with the  the following lemma which will allow us to group non-expanding components according to their closure.   We say that a set $A \subseteq X$ is \textit{closed} if $A = [A]$.
\begin{lemma}
\label{lem:countcover}
Let $A\subseteq X$ be a $2$-linked, closed, non-expanding set. Then there is a randomized $\operatorname{poly}(n) \cdot{\epsilon^{-2}} \ln(1/\delta) \cdot 2^{O\left(\frac{|A|\log^2 d}{d}\right)}$-time algorithm that outputs an $\eps$-relative approximation to the number of $2$-linked $B \subseteq A$ such that $N(B) = N(A)$ with probability at least $1 - \delta$.
\end{lemma}
\begin{proof}
Let $W=N(A)$, $w=|W|$ and $a=|[A]|$. Let 
\[\mathcal{D} := \{B \subseteq A~|~N(B) = W~\text{and}~B~\text{is $2$-linked}\}\]
 be the set whose size we would like to estimate. By Corollary~\ref{cor:cover} there exists a set $A'\in \mathcal{D}$ of size at most
\[
2\frac{a}{d} \ln d+ 2\frac{w}{d}+ 2(w-a)= O\left(\frac{a\log^2 d}{d}\right)\, .
\]
It follows that 
\[
|\mathcal{D}| \geq 2^{a- O\left(\frac{a \log^2 d}{d}\right)}.
\] 
Indeed every superset $B \supseteq A'$ satisfies $N(B) = W$ and $B$ is $2$-linked. The first property is clear, since $N(B) \supseteq N(A'') = W$. The second property holds because every vertex in $B$ is either in $A'$ or at a distance $2$ from some vertex in $A'$, which is itself $2$-linked. It follows that $|\mathcal{D}|$ can be estimated to relative error $\epsilon$ by sampling 
\[ \frac{1}{\epsilon^2}\ln (1/\delta)\cdot 2^{O\left(\frac{a \log^2 d}{d}\right)} \] 
subsets of $A$.
\end{proof}

\subsection*{The algorithm}

For a closed, $2$-linked subset $A \subseteq X$, let us denote 
\[
\mathcal{D}(A) := \#\{B\subseteq A~|~B~\text{is}~2\text{-linked},~N(B) = N(A)\}.
\]
We now define  an algorithm with inputs a graph $G$ on $n$ vertices and an accuracy parameter $\epsilon>0$ as follows. Let $L:=\lceil \frac{d}{2\log^2 d} \log(2n/\eps) \rceil$. If $d\leq \sqrt{n}$, the algorithm is as follows:
\begin{enumerate}
\item List all vectors $(a_1,\ldots,a_{\ell})$ of positive integers such that $\ell \leq n/d$ and $\sum a_i \leq n$.

\item For each vector $(a_1,\ldots,a_{\ell})$ from Step 1, list all sets $\{A_1,\ldots, A_\ell\}$ such that the $N(A_i)$'s are pairwise disjoint and $A_i \in \mathcal{G}'(v_i,a_i)$ for some $v_i\in X$ for each $i$.

\item For each set $\{A_1,\ldots, A_\ell\}$ from Step 2, compute $\tilde{\mathcal{D}}(A_i)$, which is a $(\epsilon/2n)$-relative approximation to $\mathcal{D}(A_i)$ using Lemma~\ref{lem:countcover}, setting $\delta = (1/3)\cdot 2^{-n^2}$ for all $i$.

\item For each set $A=\{A_1,\ldots, A_\ell\}$ from Step 2, let $Y_A =  Y \setminus N(\cup_{i = 1}^{\ell}A_i)$, $X_A = X \setminus N^2(\cup_{i = 1}^{\ell}A_i)$, $G_A=G[X_A\cup Y_A]$ and compute ${\Xi}^{X_A}_{G_A}(L)$.

\item Output 
\[
\sum_{\ell = 0}^{n/d} \sum_{\substack{A_1,\ldots,A_\ell~\text{compatible}\\ \forall i,~A_i~\text{non-expanding} \\ 2\text{-linked, closed}}}\left(\prod_{i =1}^{\ell}\tilde{\mathcal{D}}(A_i)\right) \cdot \left(2^{|Y| - \sum_{i = 1}^{\ell}N(A_i)}\cdot {\Xi}^{X_A}_{G_A}(L)\right)\, .
\]
\end{enumerate}

If $d>\sqrt{n}$, the algorithm is to run Steps 1-3 above and output
\begin{align}\label{eq:dlargeapprox}
\sum_{\ell = 0}^{n/d} \sum_{\substack{A_1,\ldots,A_\ell~\text{compatible}\\ \forall i,~A_i~\text{non-expanding} \\ 2\text{-linked, closed}}}\left(\prod_{i =1}^{\ell}\tilde{\mathcal{D}}(A_i)\right) \cdot \left(2^{|Y| - \sum_{i = 1}^{\ell}N(A_i)}\right)\, .
\end{align}

\subsection*{Proof of Theorem~\ref{thmMain}}

We first prove the correctness of the algorithm: for any $c>0$ and $\eps=n^{-c}$, the output is an an $\eps$-relative approximation to $i(G)$. As before we let $X, Y$ denote the vertex classes of $G$. Suppose first that $d\leq\sqrt{n}$. We then have 
\begin{align*}
i(G)&=2^{|Y|} \cdot \sum_{t = 0}^{n/d} \sum_{\substack{A_1,\ldots,A_t~\text{compatible} \\ A_i~2\text{-linked}~\forall i}} \left(\prod_{i = 1}^t 2^{-N(A_i)}\right)\\
& = 2^{|Y|} \cdot \sum_{t = 0}^{n/d} \sum_{\ell = 0}^t \sum_{\substack{A_1,\ldots,A_\ell~\text{compatible}\\ \forall i,~A_i~\text{non-expanding}, \\ 2\text{-linked}}} \left( \prod_{i = 1}^{\ell}2^{-N(A_i)}\cdot \sum_{\substack{A_{\ell+1},\ldots,A_t \\ \subseteq X \setminus N^2(\cup_{j = 1}^{\ell}A_j)\\ \text{compatible}\\ \forall i~ A_i~\text{expanding},\\ 2\text{-linked}}}\prod_{i = \ell+1}^t 2^{-N(A_i)} \right)\\
& = 2^{|Y|} \cdot \sum_{t = 0}^{n/d} \sum_{\ell = 0}^t \sum_{\substack{A_1,\ldots,A_\ell~\text{compatible}\\ \forall i,~A_i~\text{non-expanding} \\ 2\text{-linked, closed}}} \sum_{\substack{B_1,\ldots,B_{\ell} \\ \forall i,~B_i \subseteq A_i, \\ N(B_i) = N(A_i), \\ B_i~2\text{-linked}}}\left(\prod_{i = 1}^{\ell}2^{-N(A_i)}\cdot\sum_{\substack{A_{\ell+1},\ldots,A_t \\ \subseteq X \setminus N^2(\cup_{j = 1}^{\ell}A_j)\\ \text{compatible}\\ \forall i~ A_i~\text{expanding},\\ 2\text{-linked}}} \prod_{i = \ell+1}^t2^{-N(A_i)}\right) \\
& =  \sum_{\ell = 1}^{n/d} \sum_{\substack{A_1,\ldots,A_\ell~\text{compatible}\\ \forall i,~A_i~\text{non-expanding} \\ 2\text{-linked, closed}}} \sum_{\substack{B_1,\ldots,B_{\ell} \\ \forall i,~B_i \subseteq A_i, \\ N(B_i) = N(A_i), \\ B_i~2\text{-linked}}}\left(2^{|Y| - \sum_{i = 1}^{\ell}N(A_i)}\cdot \sum_{t = 1}^{n/d - \ell} \sum_{\substack{A_{1}',\ldots,A_t' \\ \subseteq X \setminus N^2(\cup_{j = 1}^{\ell}A_j)\\ \text{compatible}\\ \forall i~ A_i'~\text{expanding},\\ 2\text{-linked}}} \prod_{i = 1}^t2^{-N(A_i')}\right) \\
&= \sum_{\ell = 1}^{n/d} \sum_{\substack{A_1,\ldots,A_\ell~\text{compatible}\\ \forall i,~A_i~\text{non-expanding} \\ 2\text{-linked, closed}}}\left(\prod_{i =1}^{\ell}\mathcal{D}(A_i)\right) \cdot \left(2^{|Y| - \sum_{i = 1}^{\ell}N(A_i)}\cdot {\Xi}^{X_A}_{G_A}\right) \\
& = (1 \pm \epsilon)\sum_{\ell = 1}^{n/d} \sum_{\substack{A_1,\ldots,A_\ell~\text{compatible}\\ \forall i,~A_i~\text{non-expanding} \\ 2\text{-linked, closed}}}\left(\prod_{i =1}^{\ell}\tilde{\mathcal{D}}(A_i)\right) \cdot \left(2^{|Y| - \sum_{i = 1}^{\ell}N(A_i)}\cdot{\Xi}^{X_A}_{G_A}(L)\right).
\end{align*}
For the last equality we used that by Remark~\ref{rem:subg}, we may apply Lemma~\ref{lem:xihat} to $G_A$, and so ${\Xi}^{X_A}_{G_A}(L)$ is an $\eps$-relative approximation to ${\Xi}^{X_A}_{G_A}$. Observe that there are at most $2^{n^2}$ many choices of nonexpanding $A_1,\ldots,A_{\ell}$. So by union bounding over all these choices, the last summation is exactly what the algorithm outputs with probability at least $1 - \delta \cdot 2^{n^2} = 2/3$, we have the required approximation guarantee.
If $d\geq \sqrt{n}$, then we note that since the polymers of $G_A$ are a subset of the the polymers of $G$, we have 
by~\eqref{eqn:pfbound} that $\ln \Xi_{G_A}^{X_A}\leq \ln \Xi_{G}^{X}\leq 2^{-\Omega(\log^2 n)}$. It follows that $1$ is trivially an $\eps$-relative approximation to $\Xi^{X_A}_{G_A}$ (recall that $\eps=n^{-c}$) and so~\eqref{eq:dlargeapprox} is an $\eps$-relative approximation to $i(G)$.

We now show that if $\eps=n^{-c}$, with $c>0$ fixed,   the above algorithm runs in time $2^{O\left(\frac{\log^3 d}{d}n \right)}$. We consider the algorithm step by step. 

\noindent\textbf{Step 1.} For $\ell\leq n/d$, and $k\leq n$ the number of vectors $(a_1,\ldots,a_{\ell})$ of positive integers such that $\sum a_i = k$ (i.e.\ the number of ordered partitions of $k$ with $\ell$ parts) is
\[
\binom{k-1}{\ell-1}\leq\binom{n}{n/d}=2^{O\left(\frac{\log d}{d}n \right)}\, .
\]
Moreover, it is clear that the set of all such partitions can be listed in time $2^{O\left(\frac{\log d}{d}n \right)}$
and so Step 1 takes time $2^{O\left(\frac{\log d}{d}n \right)}$.

\noindent\textbf{Step 2.} Let $(a_1,\ldots,a_{\ell})$ be a vector of positive integers such that $\ell \leq n/d$ and $\sum a_i \leq n$. We first list all tuples vertices $\{v_1, \ldots, v_\ell\}\subset X$ which takes time $\binom{n}{\ell}=2^{O\left(\frac{\log d}{d}n \right)}$. For each $\{v_1, \ldots, v_\ell\}$, we then appeal to Lemma~\ref{lem:nonexp} to output the tuple  $(\mathcal{G}'(v_1,a_1),\ldots,  \mathcal{G}'(v_\ell,a_\ell))$ in time $2^{O\left(\frac{\log^3 d}{d}n \right)}$. We note that by Lemma~\ref{lem:nonexp}  
\[
|\mathcal{G}'(v_1,a_1)\times \ldots \times \mathcal{G}'(v_\ell,a_\ell)|\leq\prod_{i=1}^\ell 2^{O\left(\frac{\log^3 d}{d}a_i \right)}= 2^{O\left(\frac{\log^3 d}{d}n \right)}\, . 
\] 
We may therefore check each element of $\mathcal{G}'(v_1,a_1)\times \ldots \times \mathcal{G}'(v_\ell,a_\ell)$ to see if it satisfies the required conditions and output the desired list in time $2^{O\left(\frac{\log^3 d}{d}n \right)}$.

\noindent\textbf{Step 3.} Given a set $\{A_1,\ldots, A_\ell\}$ from Step 2, we use Lemma~\ref{lem:countcover} to compute an $\eps/n$-relative approximation $\tilde{\mathcal{D}}(A_i)$ to $\mathcal{D}(A_i)$ for all $i$. This takes time
\[
(\eps/n)^{-2}\ln(1/\delta)2^{O\left(\frac{\log^2 d}{d}n \right)}= 2^{O\left(\frac{\log^2 d}{d}n \right)}
\] 
where we recall that $\eps=n^{-C}$ and $\delta = (1/3) \cdot 2^{-n^2}$.

\noindent \textbf{Step 4.}
 If $d<\sqrt{n}$, then given any set $A=\{A_1,\ldots, A_\ell\}$ from Step 2, we may compute ${\Xi}^{X_A}_{G_A}(L)$ in time $(n/\eps)^{O(d)}=2^{O\left(\frac{\log^2 d}{d}n \right)}$ by the algorithm in Section~\ref{sec:CEalgorithm} restricted to polymers of $G_A$. If $d\geq\sqrt{n}$, we skip Step 4. 
 
We conclude that the algorithm takes time $2^{O\left(\frac{\log^3 d}{d}n \right)}$ in total.

\section{Weighted independent sets: proof of Theorem~\ref{thmLamExpand}}
\label{secWeighted}

The proof of Theorem~\ref{thmLamExpand} will follow the same lines as that of Theorem~\ref{thmExpander}, with the main difference being that polymer weights will now be $w_{\gamma} = \frac{ \lam^{|\gamma|}} {(1+\lam)^{|N(\gamma)|}}$, generalizing the $\lam=1$ case of Theorem~\ref{thmExpander}.

We assume throughout this section that  $G$ is a $d$-regular, bipartite $\alpha$-expander with bipartition $X,Y$ of size $n$ each.

Define a polymer model with polymers consisting of the small $2$-linked, subsets of $X$ (resp. $Y$) with two polymers compatible if their union is not $2$-linked (recall that $\gamma \subset X$ is small if $| [ \gamma ] | \le n/2$).  The weight of a polymer $\gamma$ is $w_{\gamma} = \frac{ \lam^{|\gamma|}} {(1+\lam)^{|N(\gamma)|}}$.  Let $\Xi_G^X(\lam)$ be the polymer model partition function and $\nu_{G,\lam}^X$ be the corresponding Gibbs measure on collections of compatible polymers.

Theorem~\ref{thmLamExpand}  follows from the following two lemmas. Lemma~\ref{lemXilamFPTAS} below is the analogue of Lemma~\ref{lem:CEalgorithm} and Lemma~\ref{lemXilamApprox} is the analogue of Lemma~\ref{lem:polymerapprox}.
\begin{lemma}
\label{lemXilamFPTAS}
For every $\alpha >0$, there exists $C_2>0$ so that if $\lam \ge \frac{C_2 \log d}{d^{1/4}}$ then there is an FPTAS to compute $\Xi_G^X(\lam)$ and $\Xi_G^Y(\lam)$ and a polynomial-time sampling scheme for $\nu_{G,\lam}^X$ and $\nu_{G,\lam}^Y$.
\end{lemma}

As in Section~\ref{secExpander}, we define a probability measure on independent sets as a mixture of measures derived from the two polymer models.  Define the distribution $\hat\nu^X_{G,\lam}$ on $\mathcal I(G)$ as follows.
\begin{enumerate}
\item Sample a collection of compatible polymers $\Lambda$ from the measure $\nu^X_{G,\lam}$.
\item Set $I=J \cup \bigcup_{\gamma\in\Lambda}\gamma$ where $J$ is a  random subset of $Y\backslash \bigcup_{\gamma\in\Lambda}N(\gamma)$ formed by including each vertex independently with probability $\lam/(1+\lam)$.
\end{enumerate} 
Define $\hat\nu^Y_{G,\lam}$ analogously and define the mixture
\[
\hat \mu_{G,\lam}= \frac{\Xi_G^X(\lam)}{\Xi_G^X(\lam)+\Xi_G^Y(\lam)} \hat\nu^X_{G,\lam} + \frac{\Xi_G^Y(\lam)}{\Xi_G^X(\lam)+\Xi_G^Y(\lam)} \hat\nu^Y_{G,\lam}\, .
\]

\begin{lemma}
\label{lemXilamApprox}
For every $\alpha >0$, there exists $C_2>0$ so that if $\lam \ge \frac{C_2 \log d}{d^{1/4}}$ then for $n$ sufficiently large,
\[ (1+\lam)^n \left ( \Xi_G^X( \lam)+ \Xi_G^Y( \lam) \right )  \]
is an $\epsilon$-relative approximation to $Z_G(\lam)$ 
and 
\[ \| \mu_{G,\lam} - \hat\mu_{G,\lam}  \|_{TV} < \eps \]
where $\eps = \exp ( - \Omega(n))$, with the implicit constant a function of $d$.
\end{lemma}

To prove Lemmas~\ref{lemXilamFPTAS} and~\ref{lemXilamApprox} we  extend the estimates  of Sections~\ref{secContainers} and~\ref{secPolymerModels} to the more general case.  The main graph container lemma comes from the paper of Galvin and Tetali~\cite{GT06}. Let us define 
\[
\mathcal{W}_{\lambda}(v,a,w) := \sum_{\substack{A \subset X \\|[A]| \leq n/2,\\ A~\text{is $2$-linked}\\|N(A)| = w}}\lambda^{|A|}(1+\lambda)^{-w},
\]
and
\[
\beta(\lambda) := \frac{\log^2(1 + \lambda)}{\log(1 + \lambda) + \log(2d^5/\alpha)}.
\]
We use the following lemma from~\cite{GT06}.
\begin{lemma}
\label{lem:explambda}
There are constants $c_4$ and $c_5$ such that the following holds: Let $G$ be a bipartite $\alpha$-expander  and suppose $\lambda > 0$. 
If $\beta(\lambda)$ satisfies
\[
\beta(\lambda) \geq  c_4\max\left\{\frac{\log(d^5/\alpha)}{\sqrt{d}},\frac{2\log^2 d}{\alpha d}\right\}.
\]
Then 
\[
\mathcal{W}_{\lambda}(v,a,w) \leq 2^{-c_5(w - a)\beta(\lambda)}.
\]
\end{lemma}

The hypothesis of the above lemma says that $\alpha \beta(\lambda) \geq 2c_4 \cdot \frac{\log^2 d}{d}
$. However, in our application, we will also assume that that $\beta(\lambda)$ is also large enough to  ensure
\begin{equation}
\label{eqn:alphabeta}
\alpha \beta(\lambda) \geq \frac{4000}{c_5} \cdot \frac{\log^2 d}{d} \, .
\end{equation}
This may be done by assuming that $\lam \ge \frac{C \log d}{d^{1/4}}$ for a large enough constant $C_2$.

We now prove Lemma~\ref{lemXilamFPTAS}. 
\begin{proof}[Proof of Lemma~\ref{lemXilamFPTAS}]
As in the proof of Lemma~\ref{lem:CEalgorithm}, the FPTAS and polynomial-time sampling scheme will follow from verifying the Koteck{\`y}-Preiss condition for the polymer model.

The  polymer model satisfies~(\ref{eqn:KPcondition}) with  $f(\gamma)  =  c_5\alpha\ln 2 \frac{\beta(\lambda)|\gamma|}{8}$ and $g(\gamma) = c_5\alpha \ln 2 \frac{\beta(\lambda)|N(\gamma)|}{8}$ as shown by the following computation.

\begin{align*}
\sum_{\gamma' \not\sim \gamma}w_{\gamma'}e^{f(\gamma') + g(\gamma')} & \leq \sum_{v \in \gamma} \sum_{\gamma' \not \sim v}w_{\gamma'}e^{f(\gamma') + g(\gamma')}\\
& \leq |\gamma| \cdot \max_{v\in \gamma} \sum_{\gamma' \not\sim v}\lambda^{|\gamma|}(1+\lambda)^{-|N(\gamma')|}e^{f(\gamma') + g(\gamma')}\\
& \leq |\gamma| \cdot \max_{v\in \gamma} \sum_{w \geq d} \left(\sum_{t \geq (\alpha/2)w} \mathcal{W}_{\lambda}(v,w-t,w) \cdot  e^{c_5 \alpha \ln 2 \frac{\beta(\lambda) w}{4} } \right)\\
& \leq |\gamma| \cdot \max_{v\in \gamma} \sum_{w \geq d}e^{c_5 \alpha \ln 2 \frac{\beta(\lambda) w}{4} }  \left(\sum_{t \geq (\alpha/2)w } \mathcal{W}_{\lambda}(v,w-t,w) \right)\\
& \leq |\gamma| \sum_{w \geq d}2^{c_5 \alpha  \frac{\beta(\lambda) w}{4} }  \left(\sum_{t \geq (\alpha/2)w} 2^{-c_5 \beta(\lambda)t} \right)\\
& \leq 2d|\gamma| \sum_{w \geq d}2^{c_5 \alpha \frac{\beta(\lambda) w}{4}  }  2^{-c_5 \alpha \frac{\beta(\lambda) w}{2} } \\
& \leq 4d^2|\gamma|2^{-c_5 \alpha \frac{\beta(\lambda) d}{4} }\\
& \leq 4d^2 |\gamma|2^{-500 \log^2 d} \cdot 2^{-c_5\alpha \frac{\beta(\lambda)d}{8}} \\
& \leq |\gamma| \cdot c_5\alpha\frac{\beta(\lambda)}{16}.
\end{align*}

The last inequality follows because for $d\geq3$, we have $2^{500\log^2d} < 4d^2$ and for any $x,y>0$, we have $2^{-x} \leq \frac{x}{2y}$ if $x \geq 500 \log^2 y$.

Define ${\Xi}_G^X(\ell, \lambda)$ to be the exponential of the truncated cluster expansion as in~\eqref{eqn:xihat}. By the calculation of Lemma~\ref{lem:xihat}, we have for every $\ell\geq 1$,
\begin{equation}
\label{eqn:APXCE}
\left|\ln\Xi_G^X(\lambda)-  \ln {\Xi}_G^X(\ell, \lambda) \right| \le  n \cdot 2^{- 500 \frac{\log^2 d \cdot \ell}{d}}\, .
\end{equation}
In particular, if $\ell\geq \frac{d}{1000\log^2 d} \log(n/\eps)$, then 
\[
\left|\ln\Xi_G^X-  \ln {\Xi}_G^X(\ell,\lambda) \right| \le \eps \, . \qedhere
\]
\end{proof}

We now prove Lemma~\ref{lemXilamApprox}.
\begin{proof}[Proof of Lemma~\ref{lemXilamApprox}]
Consider a modified polymer model with weights  $\tilde{w}_{\gamma}(\lambda) = \lambda^{|\gamma|}(1+\lambda)^{-|N(\gamma)|}2^{c_5\alpha\frac{\beta(\lambda)}{16}}$.
The calculation in the proof of Lemma~\ref{lemXilamFPTAS} shows that the modified polymer model with weights $\tilde{w}_{\gamma}(\lambda)$ satisfies~(\ref{eqn:KPcondition}) with  $f(\gamma)  =  c_5 \alpha\ln 2 \frac{\beta(\lambda)|\gamma|}{16}$ and $g(\gamma) = c_5 \alpha \ln 2 \frac{\beta(\lambda)|N(\gamma)|}{8}$. 

Let $\tilde{\Xi}_G^X(\lambda)$ denote the modified polymer model partition function as in \eqref{eq:xitilde}. We then have  
\[
\ln \tilde{\Xi}_G^X(\lambda) - \ln \Xi_G^X(\lambda) = \ln \mathbf{E}e^{\zeta \cdot\|\mathbf{\Lambda}\|}
\]
 for $\zeta = c_5 \alpha \ln 2\frac{\beta(\lambda)}{16}$.  Summing~(\ref{eqn:tailvertex}) over all $v$, and using the fact that $|N(\gamma)| \geq d$ for every $\gamma$, we get (as in~\eqref{eqn:pfbound})
\begin{equation}
\label{eqn:pfboundlam}
\ln \tilde{\Xi}_G^X \leq \sum_v\sum_{\substack{\Gamma \in \mathcal{C}(G)\\ \Gamma \ni v }}\left|\phi(\Gamma)\prod_{\gamma \in \Gamma}\tilde{w}_{\gamma}\right| \leq n \cdot 2^{-c_5 \alpha \beta(\lambda) \frac{d}{8}}.
\end{equation}
Therefore, we have 
\begin{align*}
\ln \mathbf{E}e^{\zeta \cdot \|\mathbf{\Lambda}\|} \leq \ln \tilde{\Xi}_G^X \leq n \cdot 2^{-c_5 \alpha \beta(\lambda)\frac{d}{8}}.
\end{align*}
Fix any $\delta \geq \frac{10}{\sqrt{d}}$. By Markov's inequality, we have
\begin{align}\label{eqLambound}
\nonumber\mathbb{P}(\|\mathbf{\Lambda}\| \geq \delta n)  & \leq \exp\left(- \zeta \delta n + n\cdot 2^{-c_5 \alpha \beta(\lambda) \frac{d}{8}}\right) \\
\nonumber&\leq 2^{- c_5\alpha\frac{\beta(\lambda)}{32} \cdot \delta n}\\
&\leq 2^{-\frac{100 \log^2 d}{d} \cdot \delta n},
\end{align}
where the penultimate inequality follows because for any $x,y >0$, we have $2^{-x \cdot y} \leq \frac{x}{2\sqrt{y}}$ if $x \geq 500\log^2 y$, and therefore
\[
   2^{-c_5 \alpha \beta(\lambda) \frac{d}{8}} \leq \frac{c_5 \alpha \beta(\lambda)}{16 \sqrt{d}} \leq  \frac{\zeta \delta}{2} .
 \] 
The final inequality~\eqref{eqLambound} follows from~\eqref{eqn:alphabeta}.

As in the proof of Lemma~\ref{lem:polymerapprox}, let $\mathcal{I}_X=\{I\in \mathcal I : \nu_{G,\lam}^X(I)>0\}$ i.e.\,  the set of all $I$ such that each $2$-linked component of $I\cap X$ is small. For $\delta>0$, let 
$\mathcal{I}^\delta_X=\{I\in \mathcal I_X: |I\cap X|\leq \delta n\}$ and define $\mathcal{I}_Y$, $\mathcal{I}^\delta_Y$  similarly. As in Lemma~\ref{lem:polymerapprox}, $\mathcal{I}=\mathcal{I}_X\cup\mathcal{I}_Y$ and so
\begin{align}\label{eqzapprox}
Z_G(\lam)=\sum_{I\in \mathcal{I}_X}\lam^{|I|}+ \sum_{I\in \mathcal{I}_Y}\lam^{|I|}-\sum_{I\in \mathcal{I}_X\cap \mathcal{I}_Y }\lam^{|I|}= (1+\lam)^n(\Xi_G^X(\lam) + \Xi_G^Y(\lam))-\sum_{I\in \mathcal{I}_X\cap \mathcal{I}_Y }\lam^{|I|}\, .
\end{align}

Now, let $I$ be a random independent set chosen from the distribution $\nu_{G,\lam}^X$. 
It follows from~\eqref{eqLambound}  that for $\delta \geq \frac{10}{\sqrt{d}}$,
\begin{equation}
\label{eqn:Xlargelambda}
\mathbb{P}(|I \cap X| > \delta n)= \sum_{I\in \mathcal{I}_X\backslash\mathcal{I}_X^\delta}\frac{\lam^{|I|}}{(1+\lam)^n\Xi_G^X(\lam)} \leq 2^{-\frac{100 \log^2 d}{d} \cdot \delta n}\, .
\end{equation}

Furthermore, we have

\begin{align*}
\sum_{I\in \mathcal{I}^\delta_X\cap\mathcal{I}^\delta_Y} \frac{\lam^{|I|}}{(1+\lam)^n\Xi_G^X(\lam)}& \leq \frac{\sum_{i,j \leq \delta n}\binom{n}{i}\binom{n}{j}\lambda^{i+j}}{(1+\lambda)^n} \\
& = \frac{\left(\sum_{i \leq \delta n}\binom{n}{i}\lambda^i\right)^2}{(1 + \lambda)^n} \\
& = (1+\lambda)^{n}\mathbb{P}\left(\operatorname{Bin}\left(n, \frac{\lambda}{1 + \lambda}\right) \leq \delta n\right)^2
\end{align*}

This quantity can be made to be at most $e^{-\frac{\delta n}{16}}$ for $\delta =  \frac{ \lambda}{100(1 + \lambda)}$. We note that $\delta \geq \frac{C \log^2 d}{d^{1/4}} \geq \frac{10}{\sqrt{d}}$ for a large enough $C$, and so
\begin{align*}
\sum_{I\in \mathcal{I}_X\cap \mathcal{I}_Y }\lam^{|I|} 
&\leq \sum_{I\in \mathcal{I}_X\backslash\mathcal{I}_X^\delta}\lam^{|I|}+  \sum_{I\in \mathcal{I}_Y\backslash\mathcal{I}_Y^\delta}\lam^{|I|} +  \sum_{I\in \mathcal{I}^\delta_X\cap\mathcal{I}^\delta_Y}\lam^{|I|}\\
&\leq (1+\lam)^n(\Xi_G^X(\lam)+\Xi_G^Y(\lam))e^{-\Omega(n)}\, .
\end{align*}

By~\eqref{eqzapprox}, we conclude that

\begin{align*}
\left|\frac{Z_G(\lam)}{(1+\lam)^n(\Xi_G^X(\lam)+\Xi_G^Y(\lam))}-1 \right|=e^{-\Omega(n)}\, .
\end{align*}
The bound on total variation distance follows in the same manner as in the proof of Lemma~\ref{lem:polymerapprox}. \qedhere
\end{proof}

\section*{Acknowledgements}
 WP is supported in part by NSF grant DMS-1847451. AP is supported in part by NSF grant CCF-1934915.

\providecommand{\bysame}{\leavevmode\hbox to3em{\hrulefill}\thinspace}
\providecommand{\MR}{\relax\ifhmode\unskip\space\fi MR }
\providecommand{\MRhref}[2]{%
  \href{http://www.ams.org/mathscinet-getitem?mr=#1}{#2}
}
\providecommand{\href}[2]{#2}

\end{document}